\documentclass[12pt]{amsart}
\usepackage{amsmath,amsthm,amssymb}

\def\pmatrix{\left(\begin{matrix}}
\def\endpmatrix{\end{matrix}\right)}
\def\Sp{\operatorname{Sp}}

\def\Z{{\mathbb Z}}
\def\F{{\mathbb Z}_2}
\def\C{{\mathbb C}}

\def\de{\delta}
\def\p{\partial}

\def\t{\theta}

\def\T{\Theta}
\def\e{\varepsilon}

\def\A{{\mathcal A}}
\def\M{{\mathcal M}}

\def\H{{\mathcal H}}
\def\tch#1#2{{\left[\begin{matrix}#1\\ #2\end{matrix}\right]}}
\def\tt#1#2{{\t\tch{#1}{#2}}}
\theoremstyle{plain}
\newtheorem{thm}{Theorem}

\newtheorem{prop}[thm]{Proposition}

\theoremstyle{definition}

\newtheorem{rem}[thm]{Remark}

\begin{document}
\title[The 2-point function for 3-loop superstrings]{The vanishing of two-point functions for three-loop superstring scattering amplitudes}
\author{Samuel Grushevsky}
\address{Mathematics Department, Princeton University, Fine Hall,
Washington Road, Princeton, NJ 08544, USA. }
\thanks{Research is supported in part by National Science Foundation under the grant DMS-05-55867.}
\email{sam@math.princeton.edu}
\author{Riccardo Salvati Manni}
\address{Dipartimento di Matematica, Universit\`a ``La Sapienza'',
Piazzale A. Moro 2, Roma, I 00185, Italy}
\email{salvati@mat.uniroma1.it}
\date{\today}

\begin{abstract}
In this paper we show that the two-point function for the three-loop chiral superstring measure ansatz proposed by Cacciatori, Dalla Piazza, and van Geemen \cite{CDPvG} vanishes. Our proof uses the reformulation of ansatz given in \cite{G}, theta functions, and specifically the theory of the $\Gamma_{00}$ linear system on Jacobians introduced by van Geemen and van der Geer \cite{vgvdg}.

At the two-loop level, where the amplitudes were computed by D'Hoker and Phong \cite{DHP1,DHPa,DHPb,DHPc,DHPd,DHPe}, we give a new proof of the vanishing of the two-point function (which was proven by them). We also discuss the possible approaches to proving the vanishing of the two-point function for the proposed ansatz in higher genera \cite{G,SM,CDPvG2}.
\end{abstract}
\maketitle

\section{Introduction}
The problem of computing the superstring measure explicitly for arbitrary genus of the worldsheet was begun by the work of
Green and Schwarz \cite{GS}, who performed the integration over supermoduli and gave an explicit formula in genus 1. D'Hoker and Phong in a series of papers \cite{DHP1,DHPa,DHPb,DHPc} introduced a gauge-fixing procedure and computed from first principles the genus 2 superstring measure, verifying  that it satisfied the physical constraints, eg.~the vanishing of the 1,2,3-point functions. They also proposed in \cite{DHPI,DHPII} to search for an ansatz for the superstring measure in arbitrary genus as the product of the bosonic measure and a modular form.

The ansatz for three-loop measure in this form was then proposed by Cacciatori, Dalla Piazza, and van Geemen in \cite{CDPvG}. The genus $g\le 3$ ansatze were reformulated in terms of syzygetic subspaces by the first author in \cite{G}, where an ansatz for general genus was proposed, under the assumption on holomorphicity of certain $2^r$-roots. Cacciatori, Dalla Piazza, and van Geemen in \cite{CDPvG2} give the genus 4 ansatz in terms of quadrics in the theta constants. The second author in \cite{SM} showed that the proposed ansatz is holomorphic in genus 5. Dalla Piazza and van Geemen in \cite{DPvG} proved the uniqueness of the modular form in genus 3 satisfying the factorization constraints. Morozov in \cite{M1} surveyed this work and gave an alternative proof that factorization constraints are satisfied for the ansatz; in \cite{M2} he has also investigated the 1,2,3-point functions of the proposed ansatz, proving under certain non-trivial mathematical assumption that they vanish on the hyperelliptic locus.

\smallskip
In this paper we use the techniques of theta functions, and especially the $\Gamma_{00}$ sublinear system of the linear system $|2\Theta|$ introduced by van Geemen and van der Geer \cite{vgvdg} to prove the vanishing of the 2-point function in genus 3. We also obtain a new proof of the vanishing of the 2-point function in genus 2.

\section{Notations and definitions}
We denote by $\A_g$ the moduli space of complex principally polarized abelian varieties of dimension $g$, and by $\H_g$ the Siegel upper half-space of symmetric complex matrices with positive-definite imaginary part, called period matrices. The space $\H_g$ is the universal cover of $\A_g$, with the deck group $\Sp(2g,\Z)$, so that we have $\A_g=\H_g/\Sp(2g,\Z)$ for a certain action of the symplectic group. A function $f:\H_g\to\C$ is called a (scalar) modular form of weight $k$ with respect to a subgroup $\Gamma\subset\Sp(2g,\Z)$ if
$$
 f(\gamma\circ\tau)=\det(C\tau+D)^kf(\tau)\qquad\forall\gamma\in\Gamma, \forall\tau\in\H_g,
$$
where $C$ and $D$ are the lower blocks if we write $\gamma$ as four $g\times g$ blocks.

For a period matrix $\tau\in\H_g$ the principal polarization $\T_\tau$ on the abelian variety $A_\tau:=\C^g/(\Z^g+\tau\Z^g)$ is the divisor of the theta function
$$
  \theta(\tau,z):=\sum\limits_{n\in\Z^g}\exp(\pi i (n^t\tau n+2n^t
  z)).
$$
Notice that for fixed $\tau$ theta is a function of $z\in\C^g$, and its automorphy properties under the lattice $\Z^g+\tau\Z^g$ define the bundle $\T_\tau$.

Given a point of order two on $A_\tau$, which can be uniquely
represented as $\frac{\tau\e+\de}{2}$ for $\e,\de\in \F^g$ (where $\F=\lbrace 0,1\rbrace$ is the additive group), the associated theta function with characteristic is
$$
  \tt\e\de(\tau,z):=\sum\limits_{n\in\Z^g}\exp(\pi i ((n+\e)^t\tau
  (n+\e)+ 2(n+\e)^t( z+\de)).
$$
As a function of $z$, $\tt\e\de$ is odd or even depending on whether
the scalar product $\e\cdot\de\in\F$ is equal to 1 or 0, respectively. The theta function with characteristic is the generator of the space of sections of the bundle $\T_\tau+\frac{\tau\e+\de}{2}$ (where we have implicitly identified the principally polarized abelian variety with its dual, and think of points as bundles of degree 0). Thus the square of any theta function with characteristic is a section of $2\T_\tau$, and the basis for the space of sections of this bundle is given by theta functions of the second order
$$
 \T[\e](\tau,z):=\tt\e0(2\tau,2z)
$$
for all $\e\in\F^g$. Riemann's addition formula is an explicit expression of the squares of theta functions with characteristics in this basis:
\begin{equation}\label{riem}
 \tt\e\de(\tau,z)^2=\sum\limits_{\sigma\in\F^g}(-1)^{\de\cdot\sigma} \T[\sigma](\tau,0)\T[\sigma+\e](\tau,z).
\end{equation}

Theta constants are restrictions of theta functions to $z=0$; thus all theta constants with odd characteristics vanish identically in $\tau$, while theta constants with even characteristics and all theta constants of the second order do not vanish identically. All theta constants with characteristics are modular forms of weight one half with respect to a certain normal subgroup of finite index $\Gamma(4,8)\subset\Sp(2g,\Z)$, while all theta constants of the second order are modular forms of weight one half with respect to a bigger normal subgroup $\Gamma(2,4)\supset\Gamma(4,8)$.

\smallskip
Theta constants with characteristics are not algebraically
independent, and satisfy a host of  algebraic identities, some of which follow from Riemann's addition formula. However, the theta constants of the second order are algebraically independent for $g=1,2$, and the only relation among them in genus 3 is of degree 16, and has been known classically. It is discussed in detail in \cite{vgvdg} --- here we give the explicit formula for easy reference. Indeed, a special case of Riemann's quartic addition theorem in genus 3 is the following identity for theta constants (where we suppress the argument $\tau$)
$$
 \tt{0\ 0\ 0}{0\ 0\ 0}\tt{0\ 0\ 0}{1\ 0\ 0}
     \tt{0\ 0\ 0}{0\ 1\ 0}\tt{0\ 0\ 0}{1\ 1\ 0}=
$$
$$
\tt{0\ 0\ 1}{0\ 0\ 0}\tt{0\ 0\ 1}{1\ 0\ 0}
     \tt{0\ 0\ 1}{0\ 1\ 0}\tt{0\ 0\ 1}{1\ 1\ 0}+
\tt{0\ 0\ 0}{0\ 0\ 1}\tt{0\ 0\ 0}{1\ 0\ 1}
     \tt{0\ 0\ 0}{0\ 1\ 1}\tt{0\ 0\ 0}{1\ 1\ 1}.
$$
If we denote the three terms in this relations by $r_i$, so that the relation is $r_1=r_2+r_3$, then multiplying the 4 ``conjugate'' relations $r_1=\pm r_2\pm r_3$ yields the identity
\begin{equation}\label{F}
  F:=r_1^4+r_2^4+r_3^4-2r_1^2r_2^2-2r_2^2r_3^2-2r_3^2r_1^2=0.
\end{equation}
Notice that $F$ is a polynomial of degree 8 in the squares of theta constants with characteristics, and thus by applying Riemann's addition formula (\ref{riem}) $F$ can be rewritten as a polynomial of degree 16 in theta constants of the second order. We refer to \cite{BL,I} for details on theta functions and modular forms, and the current knowledge about the ideal of relations among theta constants of the second order for $g>3$ (which is not known completely even for $g=4$).

\section{The linear system $\Gamma_{00}$}
In this section we review the definition and some facts about the linear system $\Gamma_{00}\subset|2\T|$ introduced and studied in \cite{vgvdg}. We refer to that paper for details, as well as to \cite{vg,G2,Iz} for a review and results on the importance of the linear system $\Gamma_{00}$ for the Schottky problem of characterizing Jacobians.

The linear system $\Gamma_{00}\subset |2\T|$ is defined to consist of all sections vanishing to order at least four at the origin. Since all sections of $2\T$ are even, this is equivalent to the value and the second derivatives $\p_{z_i}\p_{z_j}$ vanishing at zero. These conditions turn out be independent, when $(A_{\tau}, \Theta)$ is an indecomposable ppav (i.e.~not isomorphic to a product of lower-dimensional ppavs). In this case the matrix
$$
 \left(\begin{matrix} \T[\e_1](\tau,0)&\frac{\p \T[\e_1](\tau,0)}{\p\tau_{11}}&\ldots&\frac{\p \T[\e_1](\tau,0)}{\p\tau_{gg}}\\
  \vdots&\vdots&\vdots&\vdots\\
 \T[\e_{2^g}](\tau,0)&\frac{\p \T[\e_{2^g}](\tau,0)}{\p\tau_{11}}&\ldots&\frac{\p \T[\e_{2^g}](\tau,0)}{\p\tau_{gg}}\end{matrix}\right)
$$
has rank $\frac{g(g+1)}{2}+1$, cf\cite{Sas}
and thus (\cite{vgvdg}, proposition 1.1)
\begin{equation}\label{dimgamma00}
 \dim\Gamma_{00}=\dim |2\T|-1-\sum\limits_{1\le i\le j\le g}1=2^g-1-\frac{g(g+1)}{2}.
\end{equation}
Thus the linear system $\Gamma_{00}$ is zero for $g\ge 2$, has dimension 1 for $g=3$, and higher dimension for all other genera.\bigskip

The above description leads to a simple construction of a basis for  the space $\Gamma_{00}$.

\begin{prop}
Let $\tau_0$ be an irreducible point of $\H_g$ (i.e. corresponding to indecomponsable ppav). Denote $N:=1+\frac{g(g+1)}{2}$, and choose $\e_1,\ldots,\e_N\in\F^g$ such that the the modular form
$$
 g_{\e_1,\ldots,\e_N}(\tau):=\det\left(\begin{matrix}\T[\e_1](\tau,0)&\frac{\p \T[\e_1](\tau,0)}{\p\tau_{11}}&\ldots&\frac{\p \T[\e_1](\tau,0)}{\p\tau_{gg}}\\
   \vdots&\vdots&\vdots\\
\T[\e_N](\tau,0)&\frac{\p \T[\e_N](\tau,0)}{\p\tau_{11}}&\ldots&\frac{\p \T[\e_N](\tau,0)}{\p\tau_{gg}}\end{matrix}\right)
$$
does not vanish at $\tau_0$. Then the sections
$$ f_\e(\tau_0,\, z):=
\det\left(\begin{matrix}\T[\e_1](\tau_0,z)&\T[\e_1](\tau_0,0)&\frac{\p \T[\e_1](\tau_0,0)}{\p\tau_{11}}&\ldots&\frac{\p \T[\e_1](\tau_0,0)}{\p\tau_{gg}}\\
  \vdots&\vdots&\vdots&\vdots\\
  \T[\e_N](\tau_0,z)&\T[\e_N](\tau_0,0)&\frac{\p \T[\e_N](\tau_0,0)}{\p\tau_0\tau_{11}}&\ldots&\frac{\p \T[\e_N](\tau_0,0)}{\p\tau_{gg}}\\
 \T[\e](\tau_0,z)&\T[\e](\tau_0,0)&\frac{\p \T[\e](\tau_0,0)}{\p\tau_{11}}&\ldots&\frac{\p \T[\e](\tau_0,0)}{\p\tau_{gg}}
 \end{matrix}\right),
$$
for $\e\in\F^g\setminus\lbrace\e_1,\ldots,\e_N\rbrace$ form a basis of $\Gamma_{00}\subset|2\T_{\tau_0}|$.
\end{prop}
\begin{proof}
The proof is a simple linear algebra argument that we recall  for completeness. First note that  $ f_\e(\tau_0,\,z)$ belongs to $\Gamma_{00}$, as the determinant and all the second $z$-derivatives (equal to the first $\tau$-derivatives by the heat equation) vanish for $z=0$, as two of the columns of the matrix become identical. It thus remains to show that the functions $f_\e$ for various $\e$ are linearly independent. Indeed, recall that theta functions of the second order form a basis of sections of $2\T$, and now note that the basis element $\T[\e](\tau_0, z)$  enters only the expression  of  $f_\e(\tau_0,  z)$, and that with  non-zero coefficient  $g_{\e_1,\ldots,\e_N}(\tau_0)$.
\end{proof}

\begin{rem}
It can be shown that on the open sets $\lbrace g_{\e_1,\ldots,\e_N}(\tau)\neq 0\rbrace$ the coefficients of the basis vectors are in fact modular forms of weight
$g+1+N/2$, see \cite{GSM}. In particular when $g=3$ we have a global expression of the unique section $f(\tau, z)$ of the space $\Gamma_{00}$.
\end{rem}

\begin{rem}
Observe that if the  period matrix $\tau$ is decomposable, then the dimension of $\Gamma_{00} $ increases; however, a basis can still be constructed by using the same method.
\end{rem}

There exists another  method  for constructing elements of $\Gamma_{00}$ --- it is described in \cite{vgvdg}, and is as follows. Suppose $I$ is an algebraic relation among theta constants of the second order (in genus $g$). This is to say, suppose $I\in\C[x_{0\ldots0},\ldots, x_{1\ldots 1}]$ is a polynomial in $2^g$ variables such that for any $\tau\in\H_g$ we have $I(\T[\e](\tau))=0$. Then the function
$$
 f_I(z):=\sum\limits_{\e\in\F^g}\frac{\p I}{\p x_\e}\left(\T[0\ldots0](\tau,0),\ldots, \T[1\ldots1](\tau,0)\right)
 \, \T[\e](\tau,z)
$$
lies in $\Gamma_{00}\subset|2\T_\tau|$. Indeed, since $I$ vanishes identically on $\H_g$, by Euler's formula we have $f_I(0)=0$. Moreover, by the heat equation
$$
 2\pi i(1+\de_{j,k})\frac{\p^2 f_I}{\p z_j\p z_k}|_{z=0} =\sum\limits_{\e\in\F^g}\frac{\p I}{\p x_\e}
 \, \frac{\p \T[\e](\tau)}{\p\tau_{jk}}=\frac{\p I(\T[0\ldots0],\ldots\T[1\ldots1])}{\p\tau_{jk}},
$$
which is zero since $I$ vanishes identically on $\H_g$, and thus its derivative in any direction is also zero. In \cite{vgvdg}, proposition 1.2 it is shown that as $I$ ranges over the ideal of relations among theta constants, the functions $f_I$ generate the linear system $\Gamma_{00}$. Since for $g\ge 4$ the ideal of algebraic relations among theta constants of the second order is not completely known, for $g\ge 4$ this method does not yield a complete description of the basis of $\Gamma_{00}$. However, the geometry of these relations is intriguing, and this methods produces elements of $\Gamma_{00}$ with coefficients algebraic in theta constants, rather than involving their derivatives as well.

\section{The proposed ansatz for the superstring measure}
An ansatz for the 3-loop superstring measure was proposed in \cite{CDPvG}. The reformulation of this ansatz in terms of products of theta constants with characteristics in a syzygetic subspace given in \cite{G} is as follows. For any $i=0\ldots g$ define
\begin{equation}\label{defG}
  G_i^{(g)}\tch\e\de(\tau):=\sum\limits_{V\subset\F^{2g};\, \dim V=i} \ \ \prod\limits_{\tch\alpha\beta\in V}\tt{\e+\alpha}{\de+\beta}(\tau)^{2^{4-i}}.
\end{equation}
Notice that since any $i$-dimensional linear subspace contains zero, all products will contain $\tt\e\de$. Since all odd theta constants vanish identically, it is enough to sum over the even cosets of syzygetic $i$-dimensional subspaces containing $[\e,\de]$, see \cite{G,SM}.

To simplify notations, we write $m:=[\e,\de]\in\F^{2g}$ for characteristics and similarly write $\t_m:=\tt\e\de$. Then the proposed ansatz for the superstring measure is the product of the bosonic measure (which is a form on $\M_g$) and, for any even characteristic $m$, the expression
\begin{equation}\label{defXi}
 \Xi_m^{(g)}:=\sum\limits_{i=0}^g(-1)^i2^{\frac{i(i-1)}{2}} G_i^{(g)}[m]
\end{equation}
which is a modular form of weight 8 with respect to a subgroup of $\Sp(2g,\Z)$ conjugate to $\Gamma(1,2)$. In particular for genus 3 we have
$$
  \Xi_m^{(3)}:=G_0^{(3)}[m]-G_1^{(3)}[m]+2G_2^{(3)}[m]-8G_3^{(3)}[m].
$$
In \cite{SM} it is shown that the sum $\sum_m \Xi_m^{(3)}$ is a non-zero multiple of the modular form $F$ given by (\ref{F}), and thus vanishes identically on $\H_3$.

From definition (\ref{defG}) of the summands $G_i^{(g)}[m]$ of the measure $\Xi_m^{(g)}$ it follows that $G_i^{(g)}[m]$ is a polynomial in the squares of theta constants with characteristics for $i\le 3$, divisible by $\t_m^2(\tau)$. Since this is the only kind of summands appearing in the definition of $\Xi_m^{(g)}$ for $g\le 3$, by applying Riemann's addition formula (\ref{riem}) we get
\begin{prop}\label{expressible}
For $g\le 3$ the modular form $\Xi_m^{(g)}$ defined by (\ref{defXi}), and moreover the ratio $\Xi_m^{(g)}/\t_m^2(\tau,0)$ are both polynomials in theta constants of the second order, of degrees 16 and 14, respectively.
\end{prop}

\section{The vanishing of the 2-point function}
We recall (see \cite{DHPa} for explicit formulas) that the vanishing of the cosmological constant reduces to the identity $\sum_m \Xi_m^{(g)}=0$ (proven for the proposed ansatz for $g\le 4$ in \cite{SM}), and this also implies the vanishing of the 1-point function, while as shown in \cite{DHPe}the vanishing of the two-point function is equivalent to the vanishing of 
$$
 \sum_m \Xi_m^{(g)} S_m(a,b)^2
$$
for any points $a,b$ on the Riemann surface (thought of as embedded into its Jacobian), where $S_m$ is the Sz\"ego kernel
$$
 S_m(a,b):=\frac{\t_m(a-b)}{\t_m(0)E(a,b)},
$$
with $E$ being the prime form on the Riemann surface.

Since the prime form does not depend on $m$, it is a common factor in all summands above, and thus does not matter for the vanishing of the 2-point function, so the vanishing of the 2-point is equivalent to the vanishing of
$$
 X_2(a,b):=\sum_m \frac{\Xi_m^{(g)}(\tau)}{\t_m^2(\tau,0)} \t_m(\tau,a-b)^2
$$
where $\tau$ is the period matrix of the Jacobian $Jac(C)$ of a Riemann surface $C$, and $a,b\in C\subset Jac(C)$ are arbitrary.
We will now relate the vanishing of the 2-point function and the $\Gamma_{00}$ linear system. Set
\begin{equation}\label{X_2}
 X_2(\tau,z):=\sum_m\frac{\Xi_m(\tau)}{\t_m^2(\tau,0)}\t_m(\tau,z)^2
\end{equation}
and note that this function is a section of $|2\T_\tau|$. The vanishing of the 2-point function is then equivalent to $X_2(z)$ vanishing along the surface $C-C\subset Jac(C)$. By proposition 2.1 in \cite{vgvdg} and the subsequent remark, a section of $|2\T|$ vanishes along the surface $C-C$ if and only if it lies in $\Gamma_{00}$. We thus get
\begin{thm}\label{thm}
The 2-point function for the proposed superstring measure ansatz vanishes (for genus $g$) if and only if for the period matrix $\tau$ of any Jacobian of a Riemann surface of genus $g$ the section $X_2(\tau,z)$ of $2\T_\tau$ defined above lies in $\Gamma_{00}$.
\end{thm}

Since for $g=1,2$ the linear system $\Gamma_{00}$ is zero, the vanishing of the two-point function is equivalent to $X_2(\tau,z)$ vanishing identically in $z$ and $\tau$ for $g\le 2$. If we write out $X_2(\tau,z)$ as a linear combination of the basis for sections of $2\T_\tau$ given by theta functions of the second order
\begin{equation}\label{X2}
 X_2(\tau,z)=\sum c_\e(\tau)\T[\e](\tau,z),
\end{equation}
then $X_2$ vanishes identically if and only if each $c_\e(\tau)$ vanishes identically. This allowes us to recover the result of D'Hoker and Phong in genus 2.
\begin{prop}
The 2-point function for the proposed superstring ansatz vanishes identically for $g\le 2$.
\end{prop}
\begin{proof}
Let us apply Riemann's addition formula (\ref{riem}) to the definition (\ref{X_2}) of the two-point function to rewrite it in terms of theta functions of the second order (notice that the only term depending on $z$ is $\t_m(\tau,z)^2$, and we apply the addition formula to it as well). Notice that by proposition \ref{expressible} the coefficients $c_\e$ in (\ref{X2}) obtained in this way are explicit polynomials of degree 15 in theta constants of the second order, and since there are no algebraic relations among theta constants of the second order for $g\le 2$, one needs to verify that they all polynomials $c_\e$ are zero. This can be done on a computer (we used Maple). Note that the computation can be made easier by noting that since $X_2(\tau,z)$  has a transformation formula  with respect to the entire symplectic group, (i.e.~it is a Jacobi form)  and the coefficients $c_\e(\tau)$ are permuted under the action  a suitable subgroup $\Sp(2g,\Z)$ that acts monomially on the theta constants of the second order, it is enough to check that just one of $c_\e$ is the zero polynomial.
\end{proof}

In the case of $g=3$, recall from (\ref{dimgamma00}) that the space $\Gamma_{00}$ is one-dimensional, and by the results of \cite{vgvdg} we know that it is generated by
$$
 F_2:=\sum\limits_{\e\in\F^3}\frac{\p F}{\p x_\e}\left(\T[000](\tau,0),\ldots,\T[111](\tau,0)\right) \T[\e](\tau,z),
$$
where we recall that $F$, given by  (\ref{F}), is the only polynomial relation of degree 16 among the 8 theta constants of the second order for $g=3$.

\begin{prop}
For any $\tau\in\H_3$ the sections $F_2$ and $X_2$ of $2\T_\tau$ are proportional; more precisely $F_2=-\frac{14}{5}X_2$.
\end{prop}
\begin{proof}
We have explicit expressions for $F_2$ and $X_2$ as linear combinations of the basis of the sections of $2\Theta$ given by the second order theta functions. Thus what we need to verify is that the coefficient in $5F_2+14X_2$ of any $\T[\e](z)$ is equal to zero. This coefficient is a polynomial of degree 15 in theta constants of the second order and can be verified to be zero using Maple (since the only relation among theta constants of the second order is of degree 16, a polynomial in theta constants of the second order of degree 15 vanishes identically only if it is zero). Notice that by modularity it is again enough to verify that the coefficient of $\T[000](z)$ in $5F_2+14X_2$ is equal to zero.
\end{proof}

\begin{thm}
The 2-point function for the proposed ansatz for the 3-loop superstring measure vanishes identically.
\end{thm}
\begin{proof}
By the above proposition we see that for any $\tau\in\H_3$ the function $X_2$, being a constant multiple of $F_2$, lies in the linear system $\Gamma_{00}\subset |2\T_\tau|$. By theorem \ref{thm} this is equivalent to the identical vanishing of the two-point function.
\end{proof}

\begin{rem}
We  note  that the global section $F_2$ is proportional also to the global section
$f(\tau,\, z)$. Really we have that
$$F_2(\tau, z)=c(\tau) f (\tau,\, z)$$
 for any irreducible $\tau \in \H_g$. Moreover $c(\tau)$ results to be a modular function with respect to $\Sp(6,\Z)$  that is  regular on the set of irreducible point, so it is regular everywhere ( modular form) and hence it is a non zero constant. This identity produces eight non trivial identities expressing  each jacobian determinant  $ g_{\e_1,\ldots,\e_7}(\tau)$ as a polynomial of degree 15 in the theta constants
 $\T[\sigma](\tau,0)$

 \end{rem}

\section{Conclusion}
There are two generalizations that it is natural to try to prove.

\smallskip
First, one could ask whether the vanishing of the 2-point function can be obtained for the proposed in \cite{G} ansatz in higher genera. For genus 4 the ansatz is also given in \cite{CDPvG2} and is manifestly holomorphic in either formulation. The holomorphicity of the ansatz in genus 5 was proven in \cite{SM}, and thus it is natural to ask whether the 2-point function vanishes for $g\le 5$. By theorem \ref{thm} we know that this is equivalent to $X_2$ lying in the linear system $\Gamma_{00}$. However, already for genus 4 the geometry of the situation is much more complicated: instead of just one relation $F$ in genus 3 the ideal of relations among theta constants of the second order in genus 4 is unknown.
 and an explicit basis for $\Gamma_{00}$ is unknown for $g=4$ .

Moreover, it could be that here the fact that we are working on the moduli space of curves $\M_4$ rather than $\A_4$ plays a role --- the geometry of $\Gamma_{00}$ depends on this, see \cite{Iz}.

\smallskip
Second, one could try to prove the vanishing of the 3-point function. As shown in \cite{DHPe} for genus 2, this is equivalent to proving that the sum
$$
 \sum\limits_m \Xi_m^{(g)} S_m(a,b)S_m(b,c)S_m(c,a)
$$
vanishes. Using the explicit formula for the Sz\"ego kernel and canceling the $m$-independent factor, this is equivalent to the function
$$
 X_3(a,b,c):= \sum\limits_m \frac{\Xi_m(\tau)}{\t_m^3(\tau,0)} \t_m(a-b)\t_m(b-c)\t_m(c-a)
$$
vanishing identically for $a,b,c\in C$. However, in this case we do not know a natural function on $Jac(C)^{\times n}$ of which $X_3$ is a restriction, and there is no analog of the theory of the $\Gamma_{00}$ for more points. It seems that the identity among the third order theta functions obtained by Krichever in his proof of the trisecant conjecture (\cite{K}, formula (1.18)) may potentially be useful in reducing the vanishing of the 3-point to the vanishing of the 2-point function, but so far we have not been able to find an explicit way to do this.

\section*{Acknowledgements}
We are grateful to Eric D'Hoker and Duong Phong for introducing us to questions about the superstring scattering amplitudes and explanations regarding the conjectured properties of $N$-point functions. The computations for this paper were done using Maplesoft's Maple$^\copyright$ software.

\end{document}